\documentclass[twoside,twocolumn,prl,showpacs,floatfix]{revtex4-1}
 \pdfoutput=1
\usepackage{epsf,epsfig,amsfonts,amsmath,amsthm,amssymb,graphicx,times, braket}
\usepackage[mathletters]{ucs}
\usepackage[utf8x]{inputenc}
\usepackage[T1]{fontenc}
\usepackage{mathtools,xspace,color,stmaryrd,hyperref}

\newcommand{\protocoltypo}[2]{$\mathtt{#1}_\text{#2}$\xspace}
\newcommand{\PV}{\protocoltypo{PV}{}}

\newcommand{\QPV}[1]{\protocoltypo{QPV}{#1}}
\newcommand{\WSE}{\protocoltypo{WSE}{}}

\newcommand{\Operatortypo}[1]{\mathbb{#1}}
\newcommand{\id}{\Operatortypo{I}}
\newcommand{\F}{\Operatortypo{F}}

\newcommand{\stateset}[1]{\mathcal{#1}}
\newcommand{\C}{\stateset{C}}
\newcommand{\D}{\stateset{D}}
\newcommand{\sS}{\stateset{S}}

\newcommand{\AliceBob}[2]{\mathsf{#1}_{#2}}
\newcommand{\A}[1]{\AliceBob{A}{#1}}
\newcommand{\B}[1]{\AliceBob{B}{#1}}
\newcommand{\V}[1]{\AliceBob{V}{#1}}
\newcommand{\Prv}{\AliceBob{P}{}}
\newcommand{\M}[1]{\AliceBob{M}{#1}}

\DeclareMathOperator{\distance}{distance}

\newtheorem{Thm}{Theorem}
\newtheorem{Def}[Thm]{Definition}
\newtheorem{Prp}[Thm]{Property}
\newtheorem{Crl}[Thm]{Corollary}
\newtheorem{Lmm}[Thm]{Lemma}

\begin{document}

\title{A Tight Lower Bound for the BB84-states Quantum-Position-Verification Protocol}
\author{Jérémy Ribeiro}
\affiliation{Laboratoire Aimé Cotton, CNRS, Université Paris-Sud and ENS Cachan,  F-91405 Orsay, France}
\author{Frédéric Grosshans}
\email{frederic.grosshans@u-psud.fr}
\affiliation{Laboratoire Aimé Cotton, CNRS, Université Paris-Sud and ENS Cachan,  F-91405 Orsay, France}
\begin{abstract}
We use the entanglement sampling techniques developed by 
Dupuis, Fawzi and Wehner \cite{DFW15} to find a lower bound on the entanglement needed by a coalition of cheaters attacking the quantum position verification protocol 
using the four BB84 states \cite{CFG+10, BCF+14} (\QPV{BB84}) 
in the scenario where the cheaters have no access 
to a quantum channel but share a (possibly mixed) entangled state $\tilde{Φ}$. 
For a protocol using $n$ qubits, a necessary condition for cheating is that the 
max- relative entropy of entanglement $E_{\max}(\tilde{Φ})≥n-O(\log n)$. 
This improves previously known best lower bound by a factor $\sim4$, 
and it is essentially tight, since it is vulnerable to a teleportation based attack using $n-O(1)$ ebits of entanglement.
\end{abstract}
%03.67.Dd 	Quantum cryptography and communication security
%03.67.Hk 	Quantum communication
%03.67.Mn 	Entanglement measures, witnesses, and other characterizations
%89.70.Cf 	Entropy and other measures of information
%89.70.Hj 	Communication complexity

\maketitle

%\setlength\marginparwidth\columnwidth

%\section{Context and Previous Work}

The very first (classical) position verification (\PV) protocols 
have been distance bounding %(\DB) 
protocols, introduced in 1993 \cite{BrandsChaum94} 
to prevent  man-in-the-middle attacks. 
Based on the speed-limit $c$ 
on information propagation imposed by special relativity, 
they can only work 
when the prover $\Prv{}$ and the verifier $\V{}$ are close, 
and are useless against nearby malicious adversaries $\M{}$,
\emph{i.e.}\ when $\distance(\M{},\V{})≤\distance(\Prv,\V{})$ \cite{ABK+11}.
\PV  protocols by a coalition of distant verifiers $\{\V{i}\}$ are therefore
needed in such situation, as they allow to build localized authentication protocol,
but also many other cryptographic applications, like key distribution at a
specific place \cite{CGMO09}.
However, Chandran \emph{et al.} have shown in 2009 \cite{CGMO09} that no 
classical \PV protocol can be computationally secure against a coalition 
$\{\M i\}$ of malicious provers. They only found a protocol secure in the
bounded retrieval model.

Quantum position verification (\QPV{}) protocols appeared the next year in 
the scientific literature, with publications of three independent teams 
\cite{PatentKMSB06, KMS11, Malaney10a, Malaney10b, CFG+10, BCF+14}.
Even in the quantum case, unconditional security is unattainable \cite{BCF+14},
and a universal attack using an exponential amount of entanglement
as been found by Beigi and König \cite{BeigiKoenig11}.
To guarantee the security
of a \QPV{} protocol one 
either need a computational hypothesis \cite{Unruh14}
or a bound on the quantum entanglement shared between the cheaters 
\cite{LauLo11, BCF+14, BeigiKoenig11, TFKW13}. 

The present work is in the latter framework, where the cheating coalition
$\{\M i\}$ only has access to a limited amount of entanglement.
Despite the exponential universal attack \cite{BeigiKoenig11},
all lower bounds found so far 
have been linear \cite{BeigiKoenig11, TFKW13} 
or sublinear \cite{LauLo11, BCF+14}.  
To our knowledge, the protocol showing the best security in this
framework is the protocol using mutually unbiased bases \QPV{MUBs} 
proposed by Beigi and König in \cite{BeigiKoenig11}. 
A $n$-qubits implementation of
\QPV{MUBs} is secure against adversary holding less that $n/2$ ebits.
However, \QPV{MUBs} needs the coherent manipulation of $n$ qubits and
is therefore impossible to implement with present day technologies.

\QPV{BB84}, introduced in \cite{CFG+10,BCF+14} and defined 
below, 
is experimentally much simpler since it essentially uses quantum key distribution
components \cite{BB84,SBPC+09}, and Tomamichel \emph{et al.}\ \cite{TFKW13} have proved its
security against adversary holding less than $-\log_2 (\cos^2(π/8))⋅n\simeq 0.22845⋅n$ ebits
of entanglement. 
We improve this bound to $n - O(\log n)$ ebits. Since a teleportation-based explicit attack 
using $n-O(1)$ ebits is known \cite{KMS11,LauLo11}, this bound is tight.

We start this letter by giving some useful properties of 
the min-entropy $H_{\min}$ and 
the max- relative entropy of entanglement $E_{\max}$,
a related entanglement monotone.
Since our security proof is based on an adaptation of the 
entanglement sampling based security 
proof \cite{DFW15} of weak string erasure (\WSE) in the noisy storage model 
(NSM),
we then describe this protocol. We then show its
security the noisy entanglement model (NEM)
and use it to show the security of \QPV{BB84}.
% against 
%of adversaries sharing a state $Φ$ with a bounded max-entropy of entanglement 
%\cite{Datta09}
%$E_{\max}(Φ)<n-O(\log n)$.

In the following $\sS(A)$ is the set of quantum states of the system $A$.

\begin{Def}[min-entropy]\label{Def:Hmin}
  Let $ρ∈\sS(AB)$ be a bipartite state.
  The conditional min-entropy $H_{\min}(A|B)_{ρ}$ is
  \[
  H_{\min}(A|B)_{ρ}:=
   -\inf_{τ∈\sS(B)} \inf\left\{λ∈\mathbb R : ρ ≤ 2^{λ}\id_A⊗τ\right\}
  \]
\end{Def}

The following property shows the conditional min-entropy of a 
classical-quantum (\emph{cq}) state is essentially the logarithm of the 
probability to guess the classical part from the quantum part.
\begin{Prp}\label{Prp:Hminpguess}\cite[theorem 1]{KRS09}
 Let $ρ∈\sS(XB)$ be a \emph{cq}-state, \emph{i.e.} a state of the form
 $ρ=∑_xp_x\ket{x}\bra{x}⊗τ_x$ with $τ_x∈\sS(B) ∀x$. Then, 
 \[
   H_{\min}(X|B)_{ρ}=-\log_2 p_{\text{guess}}(X|B)_{ρ},
 \]
 where $p_{\text{guess}}(X|B)_{ρ})$ is the maximal probability of guessing the 
 value of $X$ from an optimal measurement on $B$.
\end{Prp}

The max- relative entropy of entanglement 
has been introduced by Datta \cite{Datta09} as
an entanglement monotone closely related to $H_{\min}$. 
\begin{Def}[max- relative entropy of entanglement]\label{Def:Emax}
Let $ρ∈\sS(AB)$ be a bipartite state.
Its max- relative entropy of entanglement is noted
$E_{\max}(ρ)_{A;B}$ or $E_{\max}(A;B)_{ρ}$ and is 
$$%E_{\max}(ρ)_{A;B}=
  E_{\max}(A;B)_{ρ}:=
  \inf_{\mathclap{σ∈\D}} \inf\left\{λ∈\mathbb R : ρ ≤ 2^{λ}σ\right\}$$
where $\D$ is the set of separable states of $\sS(AB)$.
\end{Def}

%We now establish a set of properties an theorems which will allow us to 
%quantify the usefulness of a arbitrary state $Φ$ for \WSE and \QPV{} protocols.

\begin{Prp}[monotony of $E_{\max}$]\cite[theorem 1]{Datta09}\label{Prp:Emaxmon}
The max- relative entropy of entanglement $E_{\max}$ 
is an entanglement monotone, \emph{i.e.} it can only decrease under 
local operations and classical communications (LOCC).
More formally,
let $Λ$ be completely positive trace preserving (CPTP) map $\sS(AB)→\sS(A'B')$
which can be achieved through LOCCs.
\[
 E_{\max}(ρ)_{A;B}≥ E_{\max}(Λ(ρ))_{A';B'}
\]
\end{Prp}

In order to establish the theorem \ref{Thm:EmaxvsHmin} linking $E_{\max}$
and $H_{\min}$, we will need the following lemma :
\begin{Lmm}\label{Lmm:sigmaItau}
Let $\D(A{:}B)⊂\sS(A,B)$ be the set of separable states,
\emph{i.e.}\ the convex hull of the set of product states $\sS(A)⊗\sS(B)$. 
For any state $σ∈\D(A{:}B)$, there exists a state $τ∈\sS(B)$ such that
\(
  σ≤\id⊗τ
\).
\end{Lmm}

\begin{proof}
Let $σ∈\D(A;B)$, there exists a mixture $\{p_i,τ_{A}^i⊗τ_{B}^i\}_i$ 
of states of $\sS(A)⊗\sS(B)$
such that
  \begin{align*}
   σ &=\sum_i p_i τ_{A}^i ⊗τ_{B}^i &&\text{since } σ∈\D(A{:}B)\\
     &≤\sum_i p_i \id_{A} ⊗τ_{B}^i &&\text{since } ∀i,τ_{A}^i≤\id_{A}\\
     &=\id_{A} ⊗\sum_i p_i τ_{B}^i = \id_A⊗τ
     	&&\text{defining } τ:=\sum_i p_i τ_{B}^i
  \end{align*}
\end{proof}

\begin{Thm}
\label{Thm:EmaxvsHmin}
For any bipartite state $ρ∈\sS(AB)$,
\[
  E_{\max}(A;B)_ρ≥-H_{\min}(A|B)_ρ
\]
\end{Thm}

\begin{proof}
  For any separable state $σ∈\D(A:B)$, there exists a state $τ∈\sS(B)$ such that
  \begin{align*}
    ρ&≤2^{E_{\max}(A;B)_ρ}σ &&\text{(from definition \ref{Def:Emax})}\\
     &≤2^{E_{\max}(A;B)_ρ}\id_A⊗τ &&\text{(lemma \ref{Lmm:sigmaItau})}
  \end{align*}
  The definition of $-H_{\min}$ as lower bound (definition \ref{Def:Hmin}) then
  implies
  \(
    H_{\min}(A|B)_ρ≤-E_{\max}(A;B)_ρ.
  \)
\end{proof}

%\section{Noisy Storage Model and Weak String Erasure}

Now that we have the relevant properties of $H_{\min}$ and $E_{\max}$, 
we study the weak string erasure (\WSE) protocol.
It was introduced, together with the noisy storage model (NSM)
by König \emph{et al.} \cite{KWW12}
to build secure bipartite protocols. 
%(oblivious transfer and bit commitment). 
The NSM is based on a technological limit imposed on quantum memories : 
after a delay $Δt$, the quantum state they can hold decoheres and becomes noisy.
In this model, a protocol is split in two phases.
A bipartite protocol involving the traditionally named Alice ($\A{}$) 
and Bob ($\B{}$)
can therefore be seen as a quadripartite protocol between two coalitions : it first involves early-Alice ($\A1$) and early-Bob ($\B1$), 
and then, after $Δt$, later-Alice ($\A2$) and later-Bob ($\B2$). 
A noisy quantum memory held by Bob is then modeled by a noisy quantum channel 
$\F : \B1→\B2$.

The \WSE protocol proposed in \cite{KWW12} can be described as follows, 
for honest Alice(s) and Bob(s):
\begin{enumerate}
  \item $\A1$ choses uniformly at random  $X^n=\{x_i\}_i$ and $Θ^n=\{θ_i\}_i$, 
    two bit strings of length $n$.
  \item $\B1$ choses uniformly at random $\tilde{Θ}^n=\{\tilde{θ}_i\}_i$, 
    a bit string of length $n$.
  \item \label{WSE:A1sends}
    $\A1$ sends to $\B1$ the quantum state 
    $\bigotimes_i\hat H^{θ_i}\ket{x_i}$, where
    $\hat H$ is the Hadamard operator and 
    $\{\ket{0}, \ket1\}$ the computational basis of a qubit.
    It is the BB84 encoding of 
    the string $X^n$ in the basis $Θ^n$.       
  \item \label{WSE:B1meas}
	$\B1$ measures the qubits in the bases $\tilde{Θ}^n$, and gets the string 
	$\tilde X^n$.
  \item \label{WSE:wait}
	Both parties wait the time $Δt$. The classical memories of Alice and Bob
	corresponds to classical channels allowing $\A1$ to send $\{X^n,Θ^n\}$ to $\A2$ and 
	$\B1$ to send $\{\tilde X^n,\tilde{Θ}^n\}$ to $\B2$.
  \item $\A2$ sends $Θ^n$ to $\B2$.
  \item \label{WSE:B2final}
	$\B2$ computes $I=\{i : θ_i=\tilde{θ}_i\}$ and $\tilde X^I=X^I$
\end{enumerate}

We are interested here by the correctness of the protocol,
but only  by its security against a dishonest Bob.%, defined as follows:
\begin{Def}\label{Def:security}
 A \WSE protocol is $λ$-secure against Bob if the probability for $\B2$
   to correctly guess the string $X^n$ is smaller than $2^{-nλ}$. More formally,
   let $\C(\A2,\B2)$ be the set of all possible states $σ_{\A2, \B2}$   
   which can be obtained at the end of the protocol if Alice follows it
   but Bob is dishonest.
   The protocol is secure for Alice if,  
   \(∀σ∈{\C(\A2,\B2)}\),
   \[\tfrac{1}{n}H_{\min}(X^n|\B2)_{σ}≥λ.\] 
\end{Def}

Instead of the $λ$-security, which ensures exponential security with $n$
as long as $λ>0$, one can also be interested in 
the $ε$-security for a fixed $n$ :
\begin{Def}
A protocol is $ε$-secure iff, 
for any possible dishonest strategy,
the probability $p_{\text{cheat}}$ for a dishonest adversary to win is
\(p_{\text{cheat}}<ε\).
\end{Def}
\begin{Lmm}\label{Lmm:lambdavsepsilon}
For a protocol like \WSE or \QPV{BB84}, where the goal of the cheater is
to guess a classical string $X^n$,
$λ$-security and $ε$-security are equivalent notions when 
\[
 ε=2^{-nλ} ⇔ λ=-\tfrac1n\log_2ε
\]
\end{Lmm}
\begin{proof}
This follows directly from the definitions and property \ref{Prp:Hminpguess}. 
\end{proof}

In the NSM model, a dishonest Bob changes the above protocol in the following 
way:
\begin{enumerate}
  \item[\ref{WSE:B1meas}.] $\B1$ performs a generalized measurement 
  on the qubits, obtaining a joint \emph{cq}-system ${CQ_1}$. 
  \item[\ref{WSE:wait}.] During the $Δt$ wait, $\B1$ 
    stores this state in his memory. 
    While the classical memory is perfect, 
    the quantum memory is described by the noisy channel $\F$, and
    $\B2$ obtains ${CQ_2}=(\id⊗\F)({CQ_1})$.
  \item[\ref{WSE:B2final}.] At the final step, 
    the global quantum state is $σ∈\sS(X^nΘ^nCQ_2)$,
    where $\A2$ holds the classical information $X^n$ and
    $\B2$ has access to the classical information $Θ^nC$, 
    as well as to the quantum information $Q_2$.
    $\B2$ tries to guess $X^n$ from $Θ^nCQ_2$ and the security
    of the protocol is measured by
    $H_{\min}(X^n|Θ^nCQ_2)_{σ}$.
\end{enumerate}

\begin{Thm}(\cite[theorem 14]{DFW15})\label{Thm:NSMWSE}
  Let Bob storage device $\F$ have a maximal fidelity, as defined in \cite{DFW15}, 
  upper bounded by $η$. 
  The \WSE protocol defined above is $λ$-secure for
   $$λ≤\tfrac12\left[γ\!\left(-1-\tfrac1n\log_2η\right)-\tfrac1n\right],$$
  where $γ$ is the function defined by 
  $$γ(h_{\min}):=
    \begin{cases}
      h_{\min} & \text{if } h_{\min}≥\tfrac12\\
      g^{-1}(h_{\min}) & \text{if }h_{\min}<\tfrac12,  
    \end{cases}$$
  $g(α):=h(α)+α-1$; $h(α):=-α\log_2 α - (1-α)\log_2(1-α)$ is the binary entropy function.
\end{Thm}
We defer the reader to \cite{DFW15} for the proof of this theorem.
We will now reformulate it in a slightly different security model,
the noisy entanglement model (NEM).

%\section{Weak String Erasure in the Noisy Entanglement Model}

In the NEM, $\A1$, $\A2$, $\B1$ and $\B2$ are actually four different persons, 
localized at different places and connected with (unlimited) classical channels
$\A1→\A2$ and $\B1→\B2$. 
The protocol \WSE is the same as described above, except that 
there is no specific $Δt$ at the step \ref{WSE:wait}. 
$\B1$ and $\B2$ also share 
a (possibly mixed) entangled state $\tilde{Φ}_{Q_1,Q_2}$ instead
of a quantum channel $\F$.
The two models are obviously related, 
since one can create a state $\tilde{Φ}$ by transmitting it
through $\F$, and one can create a channel $\F$ through teleportation, 
using $\tilde{Φ}$ and the unlimited classical channel.

We now adapt theorem \ref{Thm:NSMWSE} to NEM, 
exactly following Dupuis \emph{et al.}'s proof \cite{DFW15} until 
their corollary 11 and slightly changing it after.
As usual, we study the equivalent entangled protocol, 
where $\A1$ prepares a maximally entangled state
$\ket{Φ⁺}^{⊗n}_{AA'}$, sends the $A'^n$ half to $\B1$ 
and gives the $A^n$ half to $\A2$.
$\A2$ finds the string $X^n$ by measuring $A^n$ in the basis $Θ^n$.
\begin{Lmm}\cite[corollary 11]{DFW15}\label{Lmm:DFW15C11}
With the notations above, and $γ$ defined in theorem \ref{Thm:NSMWSE}, 
  we have
  \[
    H_{\min}(X^n|Θ^nCQ_2)_{σ}≥
    \tfrac12\left[nγ\!\left(\tfrac1nH_{\min}(A^n|CQ_2)_{σ}\right)-1\right]
  \]
\end{Lmm}

We can now show the security of \WSE in NEM :
\begin{Thm}\label{Thm:NEMWSE}
  Let the dishonest Bobs share a (possibly mixed) entangled state
  $\tilde{Φ}∈\sS(B_1,B_2)$. In the NEM, the \WSE protocol defined above is 
  λ-secure if
  \[
    λ≤\tfrac12\left[γ\!\left(-\tfrac1nE_{\max}(\tilde{Φ})\right)
                    -\tfrac1n\right]
  \]
\end{Thm}
\begin{proof}
We will look at the entanglement between $\B2$ and the other partners 
$\A1\A2\B1$. The only entanglement which exists 
at the beginning of the protocol comes from $\tilde{Φ}$. 
Then all the operations specified by the protocol,
as well as the ones allowed in the NEM,
are LOCCs according to the $\A1\A2\B1:\B2$ split. We have therefore 
\begin{align*}
  -E_{\max}(\tilde{Φ})
   &≤-E_{\max}(σ)_{A^n;CQ_2} &&\text{(property \ref{Prp:Emaxmon})}\\
   &≤H_{\min}(A^n|CQ_2)_{σ} && \text{(theorem \ref{Thm:EmaxvsHmin})},
\end{align*}
where $σ$ denotes the state shared by $\A2$ and $\B2$ just 
before their measurements.

Applying the monotonously increasing function $γ$ leads us to
\begin{align*}
  γ\!\left(-\tfrac1nE_{\max}(\tilde{Φ})\right)
   &≤γ\!\left(\tfrac1nH_{\min}(A^n|CQ_2)_{σ}\right)\\
\intertext{Lemma \ref{Lmm:DFW15C11} then gives}
   &≤\tfrac2n H_{\min}(X^n|Θ^nCQ_2)+\tfrac1n,
\end{align*}
which with some reordering and the
definition \ref{Def:security} of $λ$ concludes the proof.
\end{proof}

\begin{Crl}\label{Crl:WSEepsilon}
  Let $ε>0$. \WSE is $ε$-secure in the NEM if
  \[
    E_{\max}(\tilde{Φ})
     ≤ n-s-nh\!\left(\tfrac sn\right)
  \]
  where $s:=1-2\log_2 ε$ and $e$ is the basis of the natural logarithm.
  A slightly more stringent sufficient condition is :
  \[
    E_{\max}(\tilde{Φ})
     ≤ n-s\log_2n+s\log_2\tfrac{s}{2e}
  \]

\end{Crl}
\begin{proof}
  According to theorem \ref{Thm:NEMWSE}, the protocol is 
  $λ$-secure for
  \begin{align*}
  \tfrac12\left[γ\!\left(-\tfrac1nE_{\max}(\tilde{Φ})\right)
                    -\tfrac1n\right]
    &≥λ\\
  γ\!\left(-\tfrac1nE_{\max}(\tilde{Φ})\right)
    &≥\tfrac1n+2λ\\
    &=\tfrac1n-\tfrac2n\log_2ε
          \quad(\text{lemma \ref{Lmm:lambdavsepsilon}})\\
    &=:\tfrac{s}{n}\\
\intertext{applying $g=γ^{-1}$ to both sides leads to}
  -\tfrac1nE_{\max}(\tilde{Φ}) &≥ g\!\left(\tfrac sn\right)\\
  E_{\max}(\tilde{Φ}) &≤-n g\!\left(\tfrac sn\right)%\\&
    =n-s-nh\!\left(\tfrac sn\right)
  \end{align*}
  which gives the first inequality of the theorem.

  A straightforward study of the binary entropy functions shows that
  \(
    nh\!\left(\tfrac sn\right)≤s\log_2n-s\log_2\tfrac se.
  \)
  Substituting this expression in the above equation concludes the proof.
\end{proof}

%\section{Quantum Position Verification Security}

We have now all the elements to prove the security of a \QPV{} protocol.
For the sake of simplicity, we limit ourselves to the unidimensional case. 
In this case a \QPV{} protocol involves two verifiers $\{\V1,\V2\}$ 
and a prover $\Prv$ between them. 
The \QPV{BB84} protocol \cite{BCF+14} can be described as follows 
\begin{enumerate}
  \item $\V1$ and $\V2$ privately
     chose the strings $X^n$ and $Θ^n$.
  \item 
    $\V1$ sends to $\Prv$ the quantum state $\bigotimes_i\hat H^{θ_i}\ket{x_i}$. 
  \item $\V2$ sends $Θ^n$ to $\Prv$.
  \item $\Prv$ receives the messages of $\{\V1,\V2\}$ simultaneously. 
    He measures the qubits in the base $Θ^n$ and obtains $\tilde{X}^n={X}^n$. 
    He immediately broadcasts $\tilde X^n$ to $\{\V1,\V2\}$.
  \item $\{\V1,\V2\}$ accept $\Prv$'s position iff $\tilde X^n=X^n$ and if they
    receive this information on time
\end{enumerate}
The timing is such that $\Prv$ has to be at the right place to receive 
both the qubits and $Θ^n$, and then broadcast the measurement result to
$\V1$ and $\V2$ on time.
We refer the reader to \cite{BCF+14} 
for a precise definition of the timing and the correctness condition,
as we are mainly concerned by the cheating strategies.

We now study the security of this protocol against a coalition of
two malicious cheaters $\{\M1,\M2\}$, 
$\M1$ (resp. $\M2$) being closer to $\V1$ (resp. $\V2$) than $\Prv$ is supposed 
to be.
The timing constraints allow them a single round of classical 
communications. 
In the NEM they have access to no quantum communications, except an
initially shared bipartite state $\tilde{Φ}\in\sS({M_1,M_2})$. 
Note that an access to a quantum information channel 
of finite entanglement cost \cite{BBCW13}
can be brought in this model trough the corresponding state  $Φ$.
The possible action of the cheaters are:
\begin{enumerate}
  \item $\M1$ performs a generalized measurement on the qubits sent by $\V1$ and 
    his half $M_1$ of the state $Φ$. He gets a classical quantum system
    $C_1Q_1$ and sends $C_1$ to $M_2$
  \item Depending on $Θ^n$, $\M2$ performs a generalized measurement on 
    his half $M_2$ of the state $Φ$. 
    He obtains a $C_2Q_2$ and sends $C_2$ to $\M1$
  \item Receiving $C_{i±1}$, $\M i$ extracts 
     his best 
     guess $X_i^n$ from $C_1C_2Q_i$ and sends it to $\V i$
\end{enumerate}

This looks like an attack on \WSE in the NEM, where 
$\{\V i\}_i=\{\A i\}_i$ and $\{\M i\}_i=\{\B i\}_i$, with the supplementary
requirement that $\M1=\B1$ has also to output $X^n$.
In particular, it means that any attack on \QPV{BB84} leads to an attack on
\QPV{BB84} in NEM, leading us to our main result :
\begin{Thm}
  \QPV{BB84} is $ε$-secure if the state $Φ$ shared by $\M1$ and $\M2$ verifies
  \[
    E_{\max}(\tilde{Φ})
     ≤ n-s-nh\!\left(\tfrac sn\right)
  \]
  where $s:=1-2\log_2 ε$ and $e$ is the basis of the natural logarithm.
  A slightly more stringent sufficient condition is :
  \[
    E_{\max}(\tilde{Φ})
     ≤ n-s\log_2n+s\log_2\tfrac{s}{2e}
  \]
\end{Thm}

\begin{proof}

  Corollary \ref{Crl:WSEepsilon} ensures that \WSE is $ε$-secure in the NEM
  against adversaries using $Φ$ as resource.
  We will now prove by contradiction that \QPV{BB84} is also $ε$-secure.

  Let us suppose it is not the case: 
  $\M1$ and $\M2$ have a cheating strategy 
  winning in \QPV{BB84} 
  with probability $P_{\text{cheat}}^{\text{\QPV{}}}>ε$. 
  They can use this strategy as $\B1$ and $\B2$ in a \WSE protocol,
  without the $\M2→\M1$ communication and the final broadcasts of $X_i^n$,
  and using $X_2^n$ as guess for $X^n$.
  Their probability to cheat \WSE is
  \(
    P_{\text{cheat}}^{\text{\WSE}}=P(X^n_2=X^n)
      ≥P_{\text{cheat}}^{\text{\QPV{}}}>ε
  \):
  \WSE is not $ε$-secure,
  which is contradictory with corollary \ref{Crl:WSEepsilon}.

  Therefore, \QPV{BB84} is $ε$-secure.
\end{proof}

%\section{Conclusion and Outlook}

We have shown the security of the practical protocol \QPV{BB84} in one dimension against a 
coalition of cheaters sharing an entangled state of max- relative entropy
of entanglement $E_{\max}(Φ)≤n - O(\log n)$. This bound is the best known to
date for a \QPV{} protocol and is essentially tight for \QPV{BB84}, since 
an attack using $n-O(1)$ ebits is known \cite{KMS11,LauLo11}. 
While this method probably generalizes to 
the multidimensional case using tools from \cite{Unruh14}, as well as to other protocols, 
like \QPV{MUBs} \cite{BeigiKoenig11}
and non-Pauli variants of \QPV{BB84} \cite{KMS11,LauLo11},
it will not approach the exponential upper bound of these protocols.
This method is also useless when $\M1$ and $\M2$ have access to an unlimited
quantum channel 
(but did not use it for some reason to share entanglement 
before the protocol starts),
while the bound of \cite{TFKW13} works in this case.

\begin{acknowledgements}
We thank Christian Schaffner, Anthony Leverrier, Kaushik Chakraborty,
Omar Fawzi and Jędrzej Kaniewski for stimulating discussions.
FG specially thanks Christian Schaffner for introducing him to 
position based cryptography, and for maintaining the webpage 
\cite{SchaffnerPBQCWeb}, a precious resource to begin in this domain.
\end{acknowledgements}

\bibliography{QCrypt}

\end{document}